\documentclass[sort&compress]{elsarticle}

\usepackage[raiselinks=false,colorlinks=true,citecolor=blue]{hyperref}
\journal{arXiv}

\usepackage{IEEEtrantools}
\usepackage{qcircuit}
\usepackage{amsmath,amssymb,amsfonts,amsthm}
\newtheorem{theorem}{Theorem}
\newtheorem{corollary}[theorem]{Corollary}

\theoremstyle{definition}
\newtheorem{remark}[theorem]{Remark}

\usepackage{amsfonts}
\usepackage{tikz}
\usetikzlibrary{positioning,arrows.meta,angles,quotes}
\usepackage{flushend}

\newcommand{\bra}[1]{\langle #1 |}
\newcommand{\ket}[1]{| #1 \rangle}
\newcommand{\braket}[2]{\langle #1 | #2 \rangle}
\newcommand\ketbra[2]{\ket{#1}\!\bra{#2}}

\begin{document}

\begin{frontmatter}

\title{Classically time-controlled quantum automata: definition and properties}

\author[UNQ,ICC]{Alejandro D\'iaz-Caro}
\author[UNA]{Marcos Villagra}

\address[UNQ]{Universidad Nacional de Quilmes. Bernal, Buenos Aires, Argentina}
\address[ICC]{ICC, CONICET -- Universidad de Buenos Aires. Buenos Aires, Argentina}
\address[UNA]{Universidad Nacional de Asunci\'on. Asunci\'on, Paraguay}

\begin{abstract}
  In this paper, we introduce classically time-controlled quantum automata or
  CTQA, which is a reasonable modification of Moore-Crutchfield quantum finite
  automata that uses time-dependent evolution and a ``scheduler'' defining how
  long each Hamiltonian will run. Surprisingly enough, time-dependent evolution
  provides a significant change in the computational power of quantum automata
  with respect to a discrete quantum model. Indeed, we show that if a scheduler
  is not computationally restricted, then a CTQA could even decide the Halting
  problem. In order to unearth the computational capabilities of CTQAs we study
  the case of a computationally restricted scheduler. In particular, we showed
  that depending on the type of restriction imposed on the scheduler, a CTQA
  can (i) recognize non-regular languages with cut-point, even in the presence
  of Karp-Lipton advice, and (ii) recognize non-regular promise languages with
  bounded-error. Furthermore, we study the cutpoint-union of cutpoint languages
  by introducing a new model of Moore-Crutchfield quantum finite automata with
  a rotating tape head. CTQA presents itself as a new model of computation that
  provides a different approach to a formal study of ``classical control,
  quantum data'' schemes in quantum computing.
\end{abstract}


\end{frontmatter}

\section{Introduction}
A well-known hardware model for the design of quantum computers is the QRAM
model proposed by Knill~\cite{KnillTR96}. The idea is that {the} quantum device {is}
attached to a classical computer controlling all operations. Several
programming languages have been designed and studied using this model
(e.g.~\cite{SelingerValironMSCS06,GreenLeFanulumsdaineRossSelingerValironPLDI13,PaykinRandZdacewicPOPL17})
where the classical part constructs the circuit and the quantum part manipulates
the quantum state. This scheme is the so-called ``classical control,
quantum-data.''

To understand the capabilities and limitations of quantum computers with
classical control, it is interesting to conceptualize a formal model of quantum
computations that incorporates in some way the idea of a classical control. The
most simple model of computation currently known is the finite-state automaton,
and it is, arguably, the best model to initiate a study of new methods of
computation.

The first model of a quantum automaton with classical control was studied by
Ambainis and Watrous~\cite{AW02} and consisted of a two-way quantum automaton
with quantum and classical inner states, with the addition that the input tape
head is also classical. Ambainis and Watrous showed that for this model of
quantum automata there exists a non-regular language that can be recognized in
expected polynomial time, whereas for the same language any two-way
probabilistic automaton requires expected exponential time. Another way to
introduce classical components in quantum computations is in the context of
quantum interactive proof systems (QIP) with quantum automata as verifiers (e.g.,
\cite{Yakaryilmaz12,Yakaryilmaz13,NishimuraYamakami15}). These works showed that
having a quantum automaton interacting with a prover that can be quantum or
classical does indeed help the automaton to recognize more languages.

In all cited works of the previous paragraph, the classical control is
implemented via discrete circuits, that is, a ``program'' decides what gates to
apply to which qubits. However, a quantum computer could do more than just apply
discrete matrices. Indeed, the Schr\"odinger equation, which is the equation
governing the time-evolution of all quantum systems, is defined over a
continuous time, whose solutions are unitary operators.

In this work we present a new type of classical control where all unitary
operators of a quantum automaton depend on time in the sense that each
unitary operator is defined by the running time of a fixed Hamiltonian for each element
of the alphabet,
and their time-evolutions can be adjusted or tuned in order to assist the
automaton in its computations. In order to control the running time
defining each unitary operator, we introduce an idea of a
\emph{scheduler} that feeds the automaton with a \emph{time-schedule} specifying
for how long the Hamiltonian is let to run. We call this model \emph{classically time-controlled quantum
automata}, or CTQA.

The automaton model used for CTQAs is the so-called ``measure-once'' or
``Moore-Crutchfield'' quantum automaton~\cite{MC00}, where only one measurement
is allowed at the end of any computation. Brodsky and Pippenger~\cite{BP02}
proved that Moore-Crutchfield quantum automata are equivalent in computational
power to classical permutation automata, which is a much weaker and restricted
model of a deterministic finite-state automaton. The class of languages
recognized by Moore-Crutchfield automata includes only regular languages, and
there are many natural regular languages that do not belong to this class. For
example, the languages $\texttt L_{ab}=\{a^nb^m~|~n,m\geq 0\}$ and $\texttt
L_1=\{x1~|~x\in \{0,1\}^*\}$ are not recognized by any permutation automaton or
Moore-Crutchfield quantum automaton.

In this work we show that even though a
CTQA uses a quite restricted model of quantum automata, when time evolutions
 can be controlled by an external classical scheduler,
more languages can be recognized. In fact, we show that non-recursive languages
are recognized by CTQAs if we allow unrestricted time-schedules (Theorem
\ref{thm:Halting}). Since arbitrary time-schedules give extreme computational
power to a quantum automaton, we study the language recognition power of CTQAs
when assisted by computationally restricted schedulers. When the scheduler is
implemented via a deterministic finite-state automaton, we show that CTQAs can
recognize all regular languages (Theorem~\ref{the:scheduler-power}) and even
non-regular languages (Theorems~\ref{the:anbn} and \ref{the:bounded-error}).
Lastly, we define the concept of cutpoint-union between languages recognized by
CTQAs with cutpoint, and, by introducing a new model of quantum automata with
``circular input tape,'' which may be interesting by itself, we show that two
passes over an input tape suffice to recognize the cutpoint-union of two
languages recognized by two CTQAs, with some restrictions
(Theorem~\ref{the:union-c}). Moreover, if we consider only CTQAs with the same
scheduler, then their corresponding class of languages is closed under
cutpoint-union (Corollary \ref{cor:union-scheduler}).

The rest of this paper is organized as follows. In Section
\ref{sec:preliminaries} we introduce the notation used throughout this paper and
briefly review some relevant results from quantum automata theory. In Section
\ref{sec:definition} we present a formal definition of CTQAs together with some
basic properties. In Section~\ref{sec:restricted} we
present our results about restricted time-schedules. Then in Section~\ref{sec:operations} we show some basic operations
between languages. Finally, in Section~\ref{sec:open} we conclude
this paper and present some open problems.

A preliminary version of this paper appeared in~\cite{DiazcaroVillagraTPNC18}.
This new paper is a revised version and contains new material. Naturally, some
text overlaps the previous version.

\section{Preliminaries}\label{sec:preliminaries}
In this section, we briefly explain the notation used in the rest of this work and review some well-known definitions and results on quantum automata.

We use $\mathbb{R}$ to denote the set of real numbers and $\mathbb{C}$ the set of complex numbers. The set of all non-negative real numbers is denoted $\mathbb{R}_0^+$. The set of natural numbers including 0 is denoted $\mathbb{N}$. The set of rational numbers is denoted $\mathbb Q$ and the set of non-negative rational numbers is denoted $\mathbb{Q}_0^+$.

Given any finite set $A$, we let $\mathbb{C}^A$ be the Hilbert space generated by the finite basis $A$. Vectors from $\mathbb{C}^A$ are denoted using the \emph{ket} notation $\ket v$. An element of the dual space of $\mathbb{C}^A$ is denoted using the \emph{bra} notation $\bra v$. The inner product between two vectors $\ket v$ and $\ket u$ is denoted $\braket vu$.
 
Let $\Sigma$ be a finite alphabet, and let $\Sigma^*$ denote the set of all
strings of finite length over $\Sigma$. A string $x\in \Sigma^*$ of length $n$
can be written as $x=x_1\cdots x_n$ where each $x_i\in \Sigma$. The length of $x$
is denoted $|x|$. A language $\texttt L$ is a subset of $\Sigma^*$. The
concatenation of two languages $\texttt L_1$ and $\texttt L_2$ is denoted
$\texttt L_1 \cdot \texttt L_2$. We also let $\texttt L^*= \cup_{k\in
  \mathbb{N}} \texttt L^k$ where $\texttt L^0=\emptyset$, $\texttt L^1=\texttt L$, and for $k>1$, $\texttt L^k$ is the language $\texttt L$ concatenated with itself $k$ times.

A \emph{quantum finite automaton} (or QFA) is a 5-tuple
$M=(Q,\Sigma,\{\xi_\sigma~|~\sigma\in \Sigma\},s,A)$ where $Q$ is a finite set
of inner states, $\xi_\sigma$ is a transition superoperator\footnote{A
superoperator or quantum operator is a positive-semidefinite operation that
maps density matrices to density matrices~\cite{SY14}.} for a symbol $\sigma
\in \Sigma$, the initial inner state is $s\in Q$, and $A\subseteq Q$ is a set
of accepting states. On input $x\in \Sigma^*$, a computation of $M$ on
$x=x_1\cdots x_n$ is given by $\rho_j=\xi_{x_j}(\rho_{j-1})$, where
$\rho_0=\ket s \bra s$ and $1\leq j \leq |x|$. The most restricted model of QFA
currently known is the so-called \emph{Moore-Crutchfield QFA} or
MCQFA~\cite{MC00}. A MCQFA is a 5-tuple $M=(Q,\Sigma,\delta,s,A)$, where all
components are defined exactly in the same way as for QFAs except that the
transition function $\delta:Q\times \Sigma \times Q \to \mathbb{C}$ defines a
collection of unitary matrices $\{U_\sigma~|~\sigma \in \Sigma\}$ where
$U_\sigma$ has $\delta(q,\sigma,p)$ in the $(p,q)$-entry and each $U_\sigma$
acts on $\mathbb{C}^Q$. Physically $M$ corresponds to a closed-system based on
pure states.\footnote{Pure states are vectors in a complex Hilbert space
normalized with respect to the $\ell_2$-norm.} For any given input $w$, the
machine $M$ is initialized in the quantum state $\ket{\psi_0}=\ket{s}$ and each
step of a computation is given by $\ket{\psi_j}=U_{w_j}\ket{\psi_{j-1}}$, where
$1\leq j\leq |w|$. The probability that $M$ accepts $w$ is
$p_{A,M}(w)=\sum_{q_j\in A} |\braket{q_j}{\psi_{|w|}}|^2$. This is equivalent
to $M$ performing a single measurement of its quantum state at the end of a
computation. The class of languages recognized by MCQFAs with bounded-error is
denoted $\mathbf{MCQFA}$. Brodsky and Pippenger~\cite{BP02} showed using a
non-constructive argument that $\mathbf{MCQFA}$ coincides with the class of
languages recognized by permutation automata; see Villagra and
Yamakami~\cite{VY15} for a constructive argument of the same result. Ambainis
and Freivalds~\cite{AF98} studied quantum automata with pure states where
measurements are allowed at each step of a computation. We denote by
$\mathbf{KWQFA}$ the class of languages recognized by quantum automata with
pure states and with many measurements allowed. Ambainis and
Freivalds~\cite{AF98} showed that $\mathbf{MCQFA} \subsetneq  \mathbf{KWQFA}$
by proving that the language $\texttt{L}_{ab}=\{a\}^*\cdot \{b\}^*\notin
\mathbf{MCQFA}$. Furthermore, the language $\texttt{L}_{ab}\notin
\mathbf{1RFA}$, where $\mathbf{1RFA}$ is the class of languages recognized by
1-way reversible automata. The class of regular languages is denoted
$\mathbf{REG}$ and it is known that $\mathbf{KWQFA} \subsetneq
\mathbf{REG}$~\cite{AF98}.

\section{Definition and Basic Properties}\label{sec:definition}
\subsection{Formal Definition of CTQAs}

A \emph{classically time-controlled quantum automaton}\! (CTQA for short) is
defined as $(Q,\Sigma,\{H_\sigma\}_{\sigma\in \Sigma},\tau,s,A,R)$, where
\begin{itemize}
\item $Q$ is a finite set of inner states,
\item $\Sigma$ is an alphabet,
\item $\{H_\sigma\}_{\sigma\in\Sigma}$ is a set of bounded-energy
  time-independent\footnote{Remark that we take bounded-energy time-independent
  Hamiltonians, so if $H$ is valid, $-H$ also is.} Hamiltonians, one for each
  symbol in the alphabet,
\item $\tau:\Sigma^*\to (\mathbb{Q}_0^+)^*$ is a function called
  \emph{time-schedule},
\item $s$ is an initial inner state,
\item $A\subseteq Q$ is the set of accepting inner states, and
\item $R\subseteq Q$ is the set of rejecting inner states.
\end{itemize}

A CTQA has a single tape that contains an  input $x$ and a time-schedule $\tau(x)=(\tau_1\cdots\tau_{|x|})$ where each $\tau_i\in \mathbb{Q}_0^+$.

Given an input $x$ of length $n$, the time-schedule maps $x$ to a sequence of
$|x|$ positive rational numbers $\tau(x)=(\tau_1\cdots \tau_n)$ where each
$\tau_i$ indicates how much time the Hamiltonian $H_{x_i}$ must be allowed to
run. For example, taking $\hbar = 1$, the Hamiltonian $H_{x_i}$ running for
$\tau_i$ time gives the unitary operator $U=e^{-iH_{x_i}
\tau_i}=\sum_{k=0}^\infty\frac{(-iH_{x_i}\tau_i)^k}{k!}$ (which is the usual
derivation from the corresponding Schr\"odinger equation).  We call the
operator $U$ the \emph{transition operator of $H_{x_i}$ for time $\tau_i$}.
We may write such an operator
indexed by $x_i$ and $\tau_i$ as $U_{(x_i,\tau_i)}$ to make it clear how it is
obtained. Notice that $U_{(x_i,0)}=I$. We make the following remark on a fundamental result about time-dependent unitary operators and Hamiltonians.
\begin{remark}
  \label{rmk:Stone}
  Stone's theorem~\cite[Theorem B]{StoneAM32} states that for every family of
  unitary matrices $\{U(t)\}_{t\in\mathbb R}$ such that $U(t)$ is linear and
  continuous on $t$, there exists a Hamiltonian $H$ such that for each $t$,
  $U(t) = e^{iHt}$.  Therefore, we may omit the Hamiltonians and
  give only the unitary matrices.
\end{remark}

A CTQA $M$ starts in the quantum state $\ket{s}$ corresponding to the initial inner state $s$. At step $i$ if the machine $M$ is in the quantum state $\ket{\psi_{i-1}}$ and reads $x_i$, then the next quantum state $\ket{\psi_{i}}$ is given by
\[
\ket{\psi_i}=U_{(x_i,\tau_{i})}\ket{\psi_{i-1}}.
\]

After scanning an entire input, the machine $M$ observes the quantum state
$\ket{\psi_{n}}$ with respect to the subspaces $span(A)=\mathbb{C}^A$,
$span(R)=\mathbb{C}^R$ and $span(Q\setminus(A\cup
R))=\mathbb{C}^{Q\setminus(A\cup R)}$. If we observe a quantum state in
$span(A)$, we say that $x$ is accepted by $M$. Similarly, if we observe a
quantum state in $span(R)$, $x$ is rejected by $M$; otherwise, $M$ answers ``I
do not know.''

Let $\Pi_A$ be a projection onto the subspace $span(A)$ and let
\[
\ket{\psi_n}=U_{(x_n,\tau_n)}\cdots U_{(x_1,\tau_1)}\ket{s}.
\]
The probability that $M$ accepts $x$ is defined as
\[
p_{M,A}(x)=\bra{\psi_n} \Pi_A \ket{\psi_n}.
\]
Similarly, if $\Pi_R$ is a projection onto the subspace $span(R)$, the probability that $M$ rejects $x$ is
\[
p_{M,R}(x)=\bra{\psi_n} \Pi_R \ket{\psi_n}.
\]
Let $\lambda\in(0,1]$. A language $\texttt{L}$ is said to be \emph{recognized} or \emph{accepted} by $M$ with cutpoint $\lambda$ if
\[
\texttt L = \{x \in\Sigma^*~|~p_{M,A}(x)\geq\lambda\}.
\]

A CTQA $A$ is \emph{time-independent} if and only if there exists $c\in
\mathbb{Q}_0^+$ for any $x\in \Sigma^*$ such that the time-schedule
$\tau(x)=(c,\dots,c)$.  We may just write $U_\sigma$ for the transition
operators in a time-independent CTQA.

The class of languages recognized by CTQA with cutpoint $\lambda$ is denoted
$\mathbf{CTQ}_\lambda$. The class of languages recognized by time-independent
CTQA with cutpoint $\lambda$ is denoted
$\mathbf{ti}\mbox{-}\mathbf{CTQ}_\lambda$. 

\subsection{Basic Properties and Comparison with Other Models}

A language $\texttt L$ is said to be recognized by $M$ with isolated cutpoint
$\lambda$ if there exists a positive real number $\alpha$ such that
$p_{M,A}(x)\geq \lambda + \alpha$ for all $x\in \texttt L$ and $P_{M,R}(x)\leq
\lambda-\alpha$ for all $x\notin \texttt L$. Language recognition with isolated
cutpoint is easily described as recognition with bounded-error. Let $\epsilon
\in [0,\frac 12)$. We say that $\texttt L$ is recognized with bounded-error by
$M$ with error bound $\epsilon$ if $p_{M,A}(x)\geq 1-\epsilon$ for all $x\in
\texttt L$ and $p_{M,R}(x)\leq \epsilon$ for all $x\notin \texttt L$. The class
of languages recognized by CTQA with bounded-error in the time-dependent and
time-independent cases are denoted $\mathbf{BCTQ}$ and
$\mathbf{ti}\mbox{-}\mathbf{BCTQ}$, respectively.

\begin{theorem}
  \label{thm:time-independent}
$\mathbf{ti}\mbox{-}\mathbf{BCTQ}=\mathbf{MCQFA}$.
\end{theorem}
\begin{proof}
Let $A=(Q,\Sigma,\{H_\sigma\}_{\sigma\in\Sigma},q_0,A,R)$ be a time-independent
CTQA. Take $B=(Q,\Sigma,\delta',q_0,A,R)$ where
$\delta'(q,\sigma,p)=\bra{p}U_\sigma\ket{q}$. To see the other side of the
implication it suffices to see that $\delta'$ does not depend on time and thus any time-schedule works.
\end{proof}

The na\"ive definition of CTQA given in this section, allowing any arbitrary time-schedule, allows arbitrary power to CTQA, as exemplified by the following theorem.
\begin{theorem}\label{thm:Halting}
 If the time-schedule is not restricted, there exists a bounded-error CTQA deciding the Halting problem with $\epsilon=0$.
\end{theorem}
\begin{proof}
Let $\texttt{HALT}\subseteq\Sigma^*$ be the language denoting the halting problem with respect to a fixed alphabet $\Sigma$, that is, a
string $x\in \texttt{HALT}$ if and only $x$ is a reasonable encoding using an
alphabet $\Sigma$ of a Turing machine $N$ and a string $w$ such that $N$ halts
on input $w$. We construct a CTQA
$M=(Q,\Sigma,\{H_\sigma\}_{\sigma\in\Sigma},\tau,s,A,R)$ recognizing
$\texttt{HALT}$.

Let $\tau$ be a time-schedule such that if an input $x$ of $M$ is the encoding
of a Turing machine $N$ and an input $w$ for $N$, then $\tau(x)=(1,0,0,\dots,0)$
if $N$ does not halt on input $w$; otherwise, $\tau(x)=(4,0,0,\dots,0)$ if $N$
halts on input $w$. Then, define $Q=\{q_0,q_1\}$, $s=q_0$, $A=\{q_0\}$, and
$R=\{q_1\}$. Each Hamiltonian $H_\sigma$ for all $\sigma$ is defined as 
$\left(\begin{smallmatrix} 0 & \frac{\pi}2 \\ \frac{\pi}2 & 0\end{smallmatrix}\right)$, so the transition operators, for all $\sigma$, are defined as
\[
  U_{(\sigma,t)}  = \cos\left(t\cdot \frac\pi 2\right) \left(\ket{q_0}\bra{q_0}
    +\ket{q_1}\bra{q_1}\right)
  +i\sin\left(t\cdot \frac\pi 2\right)\left(\ket{q_1}\bra{q_0}-\ket{q_0}\bra{q_1}\right),
\]
and in matrix form is
\begin{equation}\label{eq:rotation}
  U_{(\sigma,t)} = R_{t\pi}=\begin{pmatrix}\cos(t\frac\pi 2) & -i\sin(t\frac\pi 2)\\ i\sin(t\frac\pi 2) & \cos(t\frac\pi 2)\end{pmatrix},
\end{equation}
where $R_{t\pi}$ denotes a $t\pi$ radian rotation about the
$x$-axis of the Bloch sphere (cf.~Figure~\ref{fig:rotation}). Note that
$U_{(\sigma,0)}=I_2$, $U_{(\sigma,1)}=-i{\sf Not}$ and $U_{(\sigma,4)}=I_2$, where ${\sf Not}$ is the quantum
negation operator.

\begin{figure*}
  \centering
\definecolor{cececec}{RGB}{236,236,236}
\definecolor{c808080}{RGB}{128,128,128}
\definecolor{cff0000}{RGB}{255,0,0}
\def \globalscale {2}
\begin{tikzpicture}
  \node (a) at (0,0) {
    \begin{tikzpicture}[y=0.80pt, x=0.80pt, yscale=-\globalscale, xscale=\globalscale, inner sep=0pt, outer sep=0pt]
      \begin{scope}[shift={(4.67722,-4.67722)}]
        \begin{scope}[shift={(-0.13363,-5.47903)}]
          \path[draw=cececec,line cap=butt,miter limit=4.00,line width=0.449pt] (72.7034,136.6464)arc(0.000:89.928:31.911505 and 31.916)arc(89.929:179.857:31.911505 and 31.916)arc(179.857:269.786:31.911505 and 31.916)arc(269.786:359.714:31.911505 and 31.916);
          \path[cm={{0.65869,-0.75241,0.72098,0.69296,(0.0,0.0)}},draw=c808080,line width=0.171pt] (-43.3994,120.9599)arc(-0.000:89.928:27.070024 and 31.912)arc(89.928:179.857:27.070024 and 31.912)arc(179.857:269.785:27.070024 and 31.912)arc(269.785:359.714:27.070024 and 31.912);
          \path[cm={{0.99976,-0.02209,0.01933,0.99981,(0.0,0.0)}},draw=c808080,line width=0.209pt] (38.2079,137.6203) ellipse (0.8939cm and 0.3206cm);
          \path[draw=black,line width=0.331pt] (70.0756,140.5505) -- (11.8049,132.8179) .. controls (70.0756,140.5505) and (11.8049,132.8179) .. (70.0756,140.5505) -- cycle;
          \path[draw=black,line width=0.186pt] (53.8025,126.0710) -- (27.9776,147.2427) .. controls (53.8025,126.0710) and (27.9776,147.2427) .. (53.8025,126.0710) -- cycle;
          \path[draw=black,line width=0.331pt] (41.4368,166.7442) -- (40.1970,107.2405) .. controls (41.4368,166.7442) and (40.1970,107.2405) .. (41.4368,166.7442) -- cycle;
          \path[cm={{0.9915,0.13014,-0.13262,0.99117,(0.0,0.0)}},draw=c808080,line width=0.210pt] (58.7476,130.3788) ellipse (0.3526cm and 0.8927cm);
          \path[draw=cff0000,line cap=butt,miter limit=4.00,line width=0.525pt] (40.7495,136.5888) -- (24.9199,114.0576) .. controls (40.7495,136.5888) and (24.9199,114.0576) .. (40.7495,136.5888) -- cycle;
          \path[line width=0.212pt] (22,110) node {$\ket\psi$};
          \path[line width=0.212pt] (24,150) node {$x$};
          \path[line width=0.212pt] (55,124) node {$-x$};
          \path[line width=0.212pt] (5,132) node {$-y$};
          \path[line width=0.212pt] (75,142) node {$y$};
          \path[line width=0.212pt] (42,100) node {$z=\ket 0$};
          \path[line width=0.212pt] (41,172) node {$-z=\ket 1$};
          \path[draw=cff0000,fill=cff0000,line cap=butt,miter limit=4.00,line width=0.525pt] (24.8058,113.9069) circle (0.0141cm);
        \end{scope}
      \end{scope}
    \end{tikzpicture}
  };
  \node (b) at (7,0) {
    \begin{tikzpicture}[y=0.80pt, x=0.80pt,yscale=-\globalscale, xscale=\globalscale,  inner sep=0pt, outer sep=0pt]
      \begin{scope}[shift={(55.75149,15.49702)}]
        \path[draw=cececec,line cap=butt,miter limit=4.00,line width=0.449pt] (106.1838,112.2589)arc(0.000:89.929:31.911505 and 31.916)arc(89.929:179.857:31.911505 and 31.916)arc(179.857:269.786:31.911505 and 31.916)arc(269.786:359.714:31.911505 and 31.916);
        \path[cm={{0.65869,-0.75241,0.72098,0.69296,(0.0,0.0)}},draw=c808080,line width=0.171pt] (-2.5719,130.0969)arc(-0.000:89.928:27.070023 and 31.912)arc(89.928:179.857:27.070023 and 31.912)arc(179.857:269.785:27.070023 and 31.912)arc(269.785:359.714:27.070023 and 31.912);
        \path[cm={{0.99976,-0.02209,0.01933,0.99981,(0.0,0.0)}},draw=c808080,line width=0.209pt] (72.1537,113.9783) ellipse (0.8939cm and 0.3206cm);
        \path[draw=black,line width=0.331pt] (103.5560,116.1630) -- (45.2854,108.4304) .. controls (103.5560,116.1630) and (45.2854,108.4304) .. (103.5560,116.1630) -- cycle;
        \path[draw=black,line width=0.186pt] (87.2830,101.6835) -- (61.4580,122.8552) .. controls (87.2830,101.6835) and (61.4580,122.8552) .. (87.2830,101.6835) -- cycle;
        \path[draw=black,line width=0.331pt] (74.9173,142.3567) -- (73.6775,82.8530) .. controls (74.9173,142.3567) and (73.6775,82.8530) .. (74.9173,142.3567) -- cycle;
        \path[cm={{0.9915,0.13014,-0.13262,0.99117,(0.0,0.0)}},draw=c808080,line width=0.210pt] (88.6983,101.8416) ellipse (0.3526cm and 0.8927cm);
        \path[draw=cff0000,line cap=butt,miter limit=4.00,line width=0.464pt] (74.2181,112.2061) -- (44.5689,102.7999) .. controls (74.2181,112.2061) and (44.5689,102.7999) .. (74.2181,112.2061) -- cycle;
        \path[draw=cff0000,fill=cff0000,line cap=butt,miter limit=4.00,line width=0.525pt] (44.5492,102.7604) circle (0.0141cm);
        \path[line width=0.212pt] (40,101) node {$\ket\psi$};
          \path[line width=0.212pt] (58,125) node {$x$};
          \path[line width=0.212pt] (89,99) node {$-x$};
          \path[line width=0.212pt] (39,107) node {$-y$};
          \path[line width=0.212pt] (109,117) node {$y$};
        \path[line width=0.212pt] (75,75) node {$z=\ket 0$};
        \path[line width=0.212pt] (74,147) node {$-z=\ket 1$};
      \end{scope}
    \end{tikzpicture}
  };
  \draw[->] (a.east) -- (b.west) node[above,midway] {$R_{\frac\pi 4}$};
\end{tikzpicture}
\caption{$U_{(\sigma,\frac 14)}=R_{\frac \pi 4}$ rotation.}
\label{fig:rotation}
\end{figure*}
  
Therefore, if the input represents a halting Turing machine, the computation will be $I_2\ket{q_0}=\ket{q_0}$ and the accepting state $\ket{q_0}$ is observed with probability 1. If the input is a non-halting Turing machine, then the computation is $-i{\sf Not}\ket{q_0}=-i\ket{q_1}$ and the rejecting state $\ket{q_1}$ is observed with probability 1.
\end{proof}

The previous theorem shows that the expressive power of a time-schedule can be easily passed to a CTQA. In order to unearth the capabilities of CTQAs we will introduce in the following section a machine called \emph{scheduler} that takes care of computing a time-schedule. As we will see, a scheduler, besides its advice-like mechanism, is endowed with oracle access to a decidable language which will be used to construct a time-schedule.

To end this section, we make two comments on our definition of time controlled
quantum computation. First, a time-schedule, as defined, is reminiscent of
advised computation. Intuitively, since a unitary operation $U_{(\sigma,t)}$ is
applied only at discrete times, a time-schedule is similar to a unitary
transition operation $U_{(\sigma,\gamma)}$ of advised quantum automata, where
$\sigma$ is an input symbol and $\gamma$ is an input advice symbol---see
Yamakami's definition of advice~\cite{Yam14}. However, given a finite set of
advised unitary transition matrices $\{U_{(\sigma,\gamma_i)}\}_{i\geq 0}$ for
each input symbol $\sigma$, in order to simulate advice using time-schedules we
need to define a unitary operation $U_{(\sigma,t)}$ that at time $t_1$ behaves
like $U_{(\sigma,\gamma_1)}$, at time $t_2$ behaves like
$U_{(\sigma,\gamma_2)}$, and so on. This process requires a construction of a
Hamiltonian $H_\sigma$ that is an interpolation of the Hamiltonians
corresponding to each advised unitary transition operation
$U_{(\sigma,\gamma_1)}, U_{(\sigma,\gamma_2)}$, etc. 
Secondly, the automata with classical and quantum states described by Ambainis and Watrous~\cite{AW02} are also intuitively similar to our model. In particular, the automaton by Qiu \emph{et al.}~\cite{QLM15} is very close to our model. In their work, the authors construct an MCQFA with classical and quantum states, where only a single measurement, dependent on the last visited classical inner state of the automaton, is applied at the end of a computation.
In contrast to the advice case, however, our model
does not have unitary operations that are conditioned on internal or external
controls. 
This is mostly due to the fact that a CTQA cannot condition which unitary operator to use based on time. Therefore, to simulate any kind of conditioned unitary operators, we require some form of interpolation of Hamiltonians, which is likely not a trivial problem.

\section{Language Recognition with Restricted Time-Schedules}\label{sec:restricted}
\subsection{The Scheduler and its Role}

A \emph{scheduler} $S$ is defined as a pair $(D,W)$ where $D$ is a multitape 
Turing machine that halts on all inputs called a \emph{decider} and $W$ is a
multi-valued function called a \emph{writer}. Besides the decider and writer,
the scheduler $S$ includes the capability of counting the size on an input. On
input $x$ a scheduler $S$ works as depicted in Figure~\ref{fig:scheduler}:
First $S$ runs $D$ on input $x$ and outputs a bit $b$ where $b=1$ if $x$ is
accepted by $D$ or $b=0$ if $x$ is rejected by $D$. Then $S$ runs the writer $W$
on input $b$ and $0^n$, where $n=|x|$. For some constant positive integer $k$, the writer $W$
is defined using two lists of functions $\mathcal{F}=(f_1,\dots,f_k)$ and
$\mathcal{G}=(g_1,\dots,g_k)$ where each $f_i:\{0\}^*\to
\mathbb{Q}_0^+$ and $g_i:\{0\}^*\to \mathbb{Q}_0^+$. The writer $W$ on input
$b$ and $0^n$ generates as an output a time-schedule $(f_{1}(0^n)\cdots
f_{n}(0^n))$ if $b=1$ or $(g_{1}(0^n)\cdots g_{n}(0^n))$ if $b=0$.
A mechanism similar to the decider $D$ appears in quantum automata with control languages~\cite{MP06}.

\begin{figure*}
\centering
\begin{tikzpicture}[node distance=.5cm,every node/.style={transform shape}, scale=0.7]
  \node[rectangle,draw,thick,minimum width=100pt,minimum height=20pt] (W) at (0,0) {Writer};
  \node (dummy) [left=4cm of W] {};
  \node[rectangle,draw,thick,minimum width=100pt,minimum height=20pt] (D) [above=of dummy] {Decider};
  \node[rectangle,draw,thick,minimum width=100pt,minimum height=20pt] (S) [below=of dummy] {Size};
  \draw[-latex,shorten >=2pt] (D.east) -- (W.west) node[midway,above,sloped]{$1/0$};
  \draw[-latex,shorten >=2pt] (S.east) -- (W.west) node[midway,below,sloped] {$|x|$};
  \node (e) [left=5cm of dummy] {};
  \node (b) [left=2.5cm of dummy] {};
  \draw[-] (e) -- (b) node[above,midway] {x};
  \draw[-latex,shorten >=2pt] (b.west) -- (D.west) node[above] {};
  \draw[-latex,shorten >=2pt] (b.west) -- (S.west) node[above] {};
  \node (o) [right=1.5cm of W] {};
  \draw[-latex] (W.east) -- (o) node[right] {output};
  \node[draw,rectangle,minimum width=12cm,minimum height=4cm,very thick,label=Scheduler] at (-3.5,0) {};
\end{tikzpicture}
\caption{Scheduler diagram}
\label{fig:scheduler}
\end{figure*}

Let $\mathcal{C}$ be a complexity class. We denote by
$\mathcal{C}\mbox{-}\mathbf{CTQ}_\lambda$ the class of languages recognized by
CTQA with cutpoint $\lambda$, where the computational power of the decider in
the scheduler is restricted to $\mathcal{C}$. That is, a language \texttt L is
in $\mathcal{C}\mbox{-}\mathbf{CTQ}_\lambda$ if there exists a set $\texttt A\in
\mathcal{C}$ such that the decider $D$ decides $\texttt A$, and a CTQA with
scheduler $(D,W)$ recognizes \texttt L with cutpoint $\lambda$. In particular,
$\mathbf{REG}\mbox{-}\mathbf{CTQ}_\lambda$ is the class of languages recognized
by CTQAs with cutpoint $\lambda$ where the decider $D$ is a finite-state
automaton. When a CTQA is bounded-error, we write
$\mathcal{C}\mbox{-}\mathbf{BCTQ}_\epsilon$, where $\epsilon$ is the error
bound.

Note that when we restrict a scheduler, we do not impose any computational
restriction on the writer. The only restriction is that any time-schedule must
contain exclusively rational numbers. This is done deliberately for the
purposes of this work to imitate practical times that can be used in a lab to
run an experiment. Furthermore, we can also use a writer to pass on some
external information to a CTQA similar to an advice-like mechanism. We explore
the power of a writer against advice later in Theorem~\ref{the:a-simb-b}.

It is clear that a CTQA has, at least, as much computational power as the
decider in its scheduler, as stated in Theorem~\ref{the:scheduler-power} below.
Later we will show that even if a scheduler is computationally restricted, a
CTQA can recognize more languages than what is allowed by its scheduler.
\begin{theorem}\label{the:scheduler-power} For any (nonempty) complexity class
$\mathcal{C}$, $\mathcal{C} \subseteq \mathcal{C}\mbox{-}\mathbf{BCTQ}_0$.
\end{theorem} \begin{proof} We can consider the same CTQA from the proof of
  Theorem~\ref{thm:Halting}. Take a decider $D$ recognizing a language $\texttt
  L\in\mathcal{C}$. Then, we consider the scheduler $S=(D,W)$ where
  $W(n,0)=(1,0,\dots,0)$ and $W(n,1)=(4,0,\dots,0)$.
  Thus, if the decider recognise the word, the CTQA will run $U_{(\sigma,4)}=I_2$, ending in the accepting state $q_0$, while if it does not, the CTQA will run $U_{(\sigma,1)} = -i\mathsf{Not}$, ending in the rejecting state $q_1$.
\end{proof}

Let $\texttt{L}_{ab}^\lambda=\{ a^n b^m\>| \> \cos^2(\frac{\pi(n-m)}{2(n+m)})\geq\lambda\}$. Using a pumping argument, it is easy to prove that $\texttt{L}_{ab}^\lambda$ for any positive $\lambda$ is not a regular language; see~\ref{ap:non-regular} for a proof. Theorem~\ref{the:anbn} below shows that if we restrict the class of the decider to the class of languages recognized by Kondacs-Watrous quantum automata, there exists a CTQA recognizing $\texttt{L}_{ab}^\lambda$ with cutpoint $\lambda$, and thus,  $\mathbf{KWQFA}\mbox{-}\mathbf{CTQ}_\lambda \nsubseteq \mathbf{REG}$.

\begin{theorem}\label{the:anbn}
There exists a CTQA with a decider from the class of languages recognized by Kondacs-Watrous quantum automata that accepts any $x\in \texttt{L}_{ab}^\lambda$ with probability at least $\lambda$, and rejects any $x\notin \texttt{L}_{ab}^\lambda$ with probability at least $1-\lambda$. 
\end{theorem}
\begin{proof}
Let $M=(Q,\Sigma,\{H_\sigma\}_{\sigma\in\Sigma},\tau,s,A,R)$ where
$Q=\{q_0,q_1\}$, $\Sigma=\{a,b\}$, $s=q_0$, $A=\{q_0\}$, $R=\{q_1\}$,
$H_a=\left(\begin{smallmatrix} 0 & \frac{\pi}2 \\ \frac{\pi}2 & 0\end{smallmatrix}\right)$, and $H_b=-H_a$. Thus, $U_{(a,t)}$ is defined as in
Eq.~\eqref{eq:rotation}, and $U_{(b,t)}$ is its opposite:
\[
  U_{(a,t)} =
   \cos\left(t\frac\pi 2 \right)\left(\ket{q_0}\bra{q_0}+\ket{q_1}\bra{q_1}\right)
  -i\sin \left(t\frac\pi 2 \right)\left(\ket{q_0}\bra{q_1}+\ket{q_1}\bra{q_0}\right),
\]
\[
  U_{(b,t)}= \cos\left(-t\frac\pi 2 \right)\left(\ket{q_0}\bra{q_0}+\ket{q_1}\bra{q_1}\right)
  -i\sin \left(-t\frac\pi 2 \right)\left(\ket{q_0}\bra{q_1}+\ket{q_1}\bra{q_0}\right).
\]
The intuition is that $U_{(a,0)}$ is the identity operator, $U_{(a,1)}$ is the
${\sf Not}$ operator, whereas $U_{(a,t)}$ is a unitary operation ``between''
the identity and the ${\sf Not}$ operators for $t\in(0,1)$. On the other hand,
$U_{(b,t)}$ is a rotation in the opposite direction.

We define the scheduler $S$ computing $\tau$ as $S=(D,W)$ where
\begin{itemize}
\item $D$ is a finite state decider recognizing the regular language $\texttt{L}_{ab}=\{a\}^*\cdot \{b\}^*$ such that $D$ outputs $b=1$ for all strings in  $\texttt L_{ab}$  and $b=0$ otherwise, and
  \item $W$ is a writer given by
  \[
  W(n+m,b)=\left\lbrace
  \begin{array}{ll}
  (\frac{1}{n+m},\frac{1}{n+m},\dots,\frac{1}{n+m}) & \textrm{if }b=1\\
  \\
  (1,0,\dots,0) & \textrm{if }b=0.
  \end{array}
  \right.
  \]
\end{itemize}

Suppose $x\in \texttt L_{ab}^\lambda$. If $x$ is the empty string, $M$ stays in the inner state $q_0$, which is an accepting inner state. Otherwise, let $x=a^nb^m$. The scheduler runs $D$ on $x$ which this time accepts $x$, and the writer outputs $(\frac{1}{n+m},\frac{1}{n+m},\dots,\frac{1}{n+m})$. The unitary operators that $M$ uses are $U_{(a,{\frac{1}{n+m}})}^n=U_{(a,{\frac n{n+m}})}$ and $U_{(b,{\frac{1}{n+m}})}^m=U_{(a,{-\frac m{n+m}})}$. After scanning the entire input 
the quantum state of $M$ is
\begin{equation}
U_{(a,-\frac m{n+m})}U_{(a,\frac n{n+m})} \ket{q_0}
=
\cos\left(\frac{\pi(n-m)}{2(n+m)}\right)\ket{q_0}
+
i\sin\left(\frac{\pi(n-m)}{2(n+m)}\right)\ket{q_1}.
  \label{eq:state-ab}
\end{equation}
Hence, the probability of accepting the string $a^nb^m$ is
$\cos^2(\frac{\pi(n-m)}{2(n+m)})$, which is greater or equal than $\lambda$
since $x$ is a member of $L_{ab}^\lambda$.

Now suppose $x\notin \texttt L_{ab}^\lambda$. The scheduler runs $D$ on $x$, which can reject or accept $x$. If $D$ rejects, then writer outputs $(1,0,\dots,0)$ as a time-schedule for $M$. The first unitary operator that is applied is either $U_{(a,1)}$ or $U_{(b,1)}$ which is the ${\sf Not}$ operator, and for each remaining 0 in the time-schedule all unitary operators are the identity. The machine $M$ will apply ${\sf Not}$ on $\ket 0$, obtaining $\ket 1$ and then it stays in $\ket 1$. After scanning the entire input, $M$ measures its quantum state and observes $\ket 1$, thus, rejecting $x$. If $D$ accepts, then the CTQA enters the same quantum state in Eq.~\eqref{eq:state-ab}. Hence, the probability of rejecting the string $a^nb^m$ is $\sin^2(\frac{\pi(n-m)}{2(n+m)})$, which is greater or equal than $1-\lambda$
since $x$ is not a member of $L_{ab}^\lambda$.
\end{proof}

\subsection{Language Recognition Power with Restricted Schedulers}

The following theorem shows that there exists a promise language\footnote{A promise language over an alphabet $\Sigma$ is a pair $\texttt L=(\texttt L_{yes},\texttt L_{no})$ such that $\texttt L_{yes},\texttt L_{no}\subseteq \Sigma^*$ and $\texttt L_{yes}\cap \texttt L_{no}=\emptyset$. We say that a promise language $\texttt L$ is \emph{solved} by a machine $M$ with bounded-error $\epsilon$ if for all $x\in \texttt L_{yes}$ we have that $M$ accepts $x$ with probability at least $1-\epsilon$ and for all $x\in \texttt L_{no}$ we have that $M$ rejects $x$ with probability at least $1-\epsilon$ \cite{AY12}.} that can be solved by CTQAs with bounded-error, even in the case of a restricted scheduler. The only caveat is that our CTQA requires endmarkers for the beginning and end of an input string. Recall that $\mathrm{KWQFA} \subsetneq \mathrm{REG}$~\cite{AF98}. 
Let Odd$^{\geq 1}\subseteq\mathbb N$ be the set of odd natural numbers and let $\texttt{L}_{ab,\delta}=(\texttt{L}_{ab,\delta}^{\text{yes}},\texttt{L}_{ab,\delta}^{\text{no}})$, with $0<\delta< \frac{\sqrt{3}}{2}-\frac 12\approx 0.366\dots$, be a promise language over the alphabet $\Sigma=\{a,b\}$ where
\begin{align*}
  &\texttt{L}_{ab,\delta}^{\text{yes}}
  =\bigg\{a^nb^m\>:\>
       \left|\frac{n}{n+m}-\frac 1 2\right|<\delta,\left|\frac{m}{n+m}-\frac 1 2\right|<\delta,
     n,m\in\textrm{Odd}^{\geq 1}\bigg\},
  \\
  & \texttt{L}_{ab,\delta}^{\text{no}}
  =
  \left\{a^nb^m\>:\> \frac{n}{n+m}<\delta^2,\> n,m\in\textrm{Odd}^{\geq 1}\right\}
  \cup (\Sigma^*-a^*b^*) .
\end{align*}
Notice that $\texttt{L}_{ab,\delta}^{\text{yes}}\cap\texttt{L}_{ab,\delta}^{\text{no}}=\emptyset$ since
for any $a^nb^m\in\texttt{L}_{ab,\delta}^{\text{yes}}$ we have $\frac 12-\delta <\frac n{n+m} < \frac 12+\delta$ while for any $a^nb^m\in\texttt{L}_{ab,\delta}^{\text{no}}$ we have $\frac n{n+m} < \delta^2<\frac 12-\delta$.

The promise language $\texttt{L}_{ab,\delta}$ distinguishes strings that have
almost the same number of $a$'s and $b$'s from strings that are not of the form
$a^*b^*$ or the number of $a$'s is much larger than the number of $b$'s. We can
use a pumping argument to show that there is no deterministic finite-state
automaton that can distinguish $\texttt{L}_{ab,\delta}^{\text{yes}}$ from
$\texttt{L}_{ab,\delta}^{\text{no}}$. Take the string $a^pb^p \in
\texttt{L}_{ab,\delta}^{\text{yes}}$, where $p$ is at least the number of
states of any deterministic finite automaton that solves
$\texttt{L}_{ab,\delta}$. Thus, we guarantee that there is a state that is
repeating while scanning $a$'s. Take those substrings of $a$'s that start and
end in the same state and pump it enough times to make a new string
$a^{h+k}b^{p}$, where $h<p$ and $k$ is sufficiently large. This way we will
have that $a^{h+k}b^p\in \texttt{L}_{ab,\delta}^{\text{no}}$ and the automaton
will accept the string when it has to reject it.

\begin{theorem}\label{the:bounded-error}
There exists a CTQA with right and left endmarkers and a decider from the class of languages recognized by Kondacs-Watrous quantum automata that solves the promise language $\texttt{L}_{ab,\delta}$, for some $\delta$, with bounded-error $\epsilon$.
\end{theorem}
\begin{proof}
We prove that there exists $\epsilon$ such that $\texttt{L}_{ab,\delta}$ is solved by some CTQA with error upper-bounded by $\epsilon$ for some appropriately chosen $\delta$.

Let $M=(Q,\Sigma,\{H_\sigma\}_{\sigma\in\hat\Sigma},\tau,s,A,R)$ where
$Q=\{q_0,q_1,q_2,q_{acc},q_{rej}\}$, $\Sigma=\{a,b\}$, $\hat\Sigma=\Sigma\cup\{\vdash,\dashv\}$, $s=q_0$, $A=\{q_{acc}\}$, $R=\{q_{rej}\}$ and $\vdash$ is a left endmarker and $\dashv$ is a right endmarker. 
The Hamiltonians $H_a$, $H_b$, $H_\vdash$, and $H_\dashv$ are chosen such that they produce the following families of unitary operators (cf.~Remark~\ref{rmk:Stone}):
  \begin{align*}
  V_{(a,t)}
  &= \sin \left(t\pi \right)\left(\ket{q_0}\bra{q_0}+\ket{q_1}\bra{q_1}\right)
  +\cos\left(t\pi \right)\ket{q_0}\bra{q_1}-\cos\left(t\pi \right)\ket{q_1}\bra{q_0}\\
  &+\ket{q_2}\bra{q_2}+\ket{q_{acc}}\bra{q_{acc}}+\ket{q_{rej}}\bra{q_{rej}},
  \\
  V_{(b,t)}&= \sin(t\pi)(\ket{q_0}\bra{q_0}+\ket{q_2}\bra{q_2})
  +\cos(t\pi)\ket{q_0}\bra{q_2}-\cos(t\pi)\ket{q_2}\bra{q_0}\\
  &+\ket{q_1}\bra{q_1}+\ket{q_{acc}}\bra{q_{acc}}+\ket{q_{rej}}\bra{q_{rej}}.
  \end{align*}

For the left endmarker we have $V_{(\vdash,t)}=R_{t\pi}\otimes I$, where $R_{t\pi}$ acts on $\{q_0,q_{rej}\}$ and $I$ is the identity matrix acting on $\{q_1,q_2, q_{acc}\}$. For the right endmarker we have that $V_{(\dashv,t)}$ is any permutation matrix where 
\begin{align*}
  V_{(\dashv,1)}\ket{q_0} & =\ket{q_{acc}}, \\
  V_{(\dashv,1)}\ket{q_1} & =\ket{q_{rej}}, \\
  V_{(\dashv,1)}\ket{q_2} & =\ket{q_2}.
\end{align*}

We define the scheduler $S$ computing $\tau$ as $S=(D,W)$ where $D$ is a finite-state decider that recognizes the language $a^*b^*$ and $W$ is a writer given by
\begin{align*}
  W(n+m,0)& =(1,0,\dots,0) \text{ and}\\
  W(n+m,1) &=(0,\overbrace{\frac{1}{n+m},\dots,\frac{1}{n+m}}^{n\text{ times}},\overbrace{\frac{1}{n+m},\dots,\frac{1}{n+m}}^{m\text{ times}},1).
\end{align*}
Now we prove the correctness of our construction by estimating the accepting and rejecting probabilities.

Suppose that $x\in \texttt{L}_{ab,\delta}^{\text{yes}}$. In this case, the
decider $D$ accepts and the writer $W$ outputs
$(0,\frac{1}{n+m},\dots,\frac{1}{n+m},\frac{1}{n+m},\dots,\frac{1}{n+m},1)$.
The automaton starts in $\ket{q_0}$ and the operator $V_{(\vdash,t)}$ acts as
the identity on $\ket{q_0}$. Now $M$ scans $a^nb^m$. First, note that, by
straightforward trigonometric properties, it is easy to check that
$V_{(a,t)}^n=\pm V_{(a,nt)}$ and $V_{(b,t)}^m=\pm V_{(b,mt)}$, where the signs are
positive for $n=1,5,9,\dots$ and negative for $n=3,7,11,\dots$.
Also, we have that
\begin{align*}
     V_{(a,\frac{1}{n+m})}^n\ket{q_0}=\pm\sin\left(\frac{n\pi}{n+m}\right)\ket{q_0}\mp\cos\left(\frac{n\pi}{n+m}\right)\ket{q_1}.
\end{align*}

After scanning an entire input $a^nb^m$ and the right endmarker, the CTQA $M$ is in the quantum state
\begin{align}
  V_{(\dashv,1)}V_{(b,\frac{1}{n+m})}^m V_{(a,\frac{1}{n+m})}^n\ket{q_{0}}
&=\left(\pm\sin\left(\frac{m\pi}{n+m}\right)\right).\left(\pm\sin\left(\frac{n\pi}{n+m}\right)\ket{q_{acc}})\right)\nonumber\\
&\quad+\left(\mp\cos\left(\frac{m\pi}{n+m}\right)\right).\left(\pm\sin\left(\frac{n\pi}{n+m}\right)\ket{q_2}\right)\nonumber\\
&\quad+\left(\mp\cos\left(\frac{n\pi}{n+m}\right)\ket{q_{rej}}\right).
    \label{eq:state-bounded}
\end{align}

The probability that $M$ accepts input $x=a^n b^m$ in $\texttt{L}_{ab,\delta}^{\text{yes}}$ is given by
\begin{align*}
p_{M,A}(x) &= \sin^2 \left( \frac{m\pi}{n+m} \right)\cdot \sin^2 \left( \frac{n\pi}{n+m} \right)\\
     &\geq \sin^2 \left(\pi\cdot (\frac{1}{2}+\delta) \right)\cdot \sin^2 \left(\pi\cdot (\frac{1}{2}+\delta) \right)\nonumber
     = \cos^4(\pi \delta).
\end{align*}

Now suppose that $x\in \texttt{L}_{ab,\delta}^{\text{no}}$. If $x\in \Sigma^*-a^*b^*$, then the writer $W$ outputs the time schedule $(1,0,\dots,0)$ and $M$ immediately enters the rejecting state $\ket{q_{rej}}$ and it rejects with probability 1. If $x$ is of the form $a^*b^*$, then the writer outputs the time schedule $(0,\frac{1}{n+m},\dots,\frac{1}{n+m},\frac{1}{n+m},\dots,\frac{1}{n+m},1)$ and $M$ arrives at the same quantum state of Eq.~\eqref{eq:state-bounded}. Thus, the probability that $M$ rejects an input $x= a^n b^m$ is given by 
\[
  p_{M,R} (x)  =\cos^2\left( \frac{n\pi}{n+m} \right)
  \]
  Since $\frac{n}{n+m}<\delta^2<\delta\leq \frac{\sqrt 3}2-\frac 12$, the function $\cos^2\left(\frac{n\pi}{n+m}\right)$ is  decreasing, hence
  \[
  p_{M,R} (x)  =\cos^2\left( \frac{n\pi}{n+m} \right)
  \geq \cos^2\left(\delta^2\pi \right).
\]

For $\delta=\frac{1}{6}$ we have that $p_{M,A}(x)\geq \cos^4(\frac{1 \pi}{36}) > 0.56$ when $x\in
\texttt{L}_{ab,\delta}^{\text{yes}}$ and $p_{M,R}(x)\geq \cos^2(\frac{1\pi}{36}) > 0.93$ when $x\in
\texttt{L}_{ab,\delta}^{\text{no}}$. Hence, we obtain an error bound of
$\epsilon\leq 0.44$. 
\end{proof}

Let $\texttt L_{1}=\{w\cdot 1~|~w \in \{0,1\}^*\}$. The language $\texttt L_{1}$
is a regular language that is not recognized by any KWQFA~\cite{AF98}. This
language can be recognized by a $\mathbf{CTQ}_\lambda$ with a decider restricted
to a constant function. Let $\mathbf{\Sigma^*}\mbox{-}\mathbf{CTQ}_\lambda$ be
the class of languages recognized by CTQAs with cutpoint $\lambda$ where the
decider accepts any string over the alphabet $\Sigma$. Similarly, we define $\mathbf{\Sigma^*}\mbox{-}\mathbf{BCTQ}_\epsilon$ to be the class of languages recognized by CTQAs with bounded-error $\epsilon$. Note that when a decider
computes a constant function, the output of the scheduler only depends on the
length of each input string. This situation is similar to quantum automata
assisted by advice, as studied in~\cite{Yam14,VY15}.

\begin{theorem}\label{thm:One-CTQA}
$\mathbf{\Sigma^*}\mbox{-}\mathbf{BCTQ}_0 \nsubseteq \mathbf{KWQFA}$.
\end{theorem}
\begin{proof}
Since $\texttt L_1\notin\mathbf{KWQFA}$, it suffices to prove that $\texttt L_1 \in \mathbf{\Sigma^*}\mbox{-}\mathbf{BCTQ}_0$. Let
$M=(Q,\Sigma$, $\{H_\sigma\}_{\sigma\in\Sigma},\tau,s,A,R)$, where
$Q=\{q_0,q_1\}$, $\Sigma=\{0,1\}$, $s=q_0$, $A=\{q_0\}$, $R=\{q_1\}$, $H_0=0$,
and $H_1=\left(\begin{smallmatrix} 0 & \frac{\pi}2 \\ \frac{\pi}2 &
0\end{smallmatrix}\right)$ (as in Eq.~\eqref{eq:rotation}). Thus, the
transition operators are defined as $U_{(0,t)}=I_2$ and $U_{(1,t)}=R_{t\pi}$.
In particular, $U_{(0,1)}=I_2$ and $U_{(1,1)}={\sf Not}$.

  The decider of the scheduler is defined by $D(x)=1$ for any $x\in\{0,1\}^*$,
  and the writer is defined by
  \[
    W(n,b) =\left\{
      \begin{array}{ll}
        (0,0,\dots,0,1)&\textrm{ if }b=1\\
        \\
        (0,0,\dots,0,0)&\textrm{ if }b=0.
      \end{array}
    \right.
  \]
Notice that since the decider is the constant function $1$, the scheduler will always output a time-schedule with $n-1$ zeroes and a single one in the last position. Therefore, the automaton $M$ will do nothing with the $n-1$ first symbols, and it will apply $U_{(0,1)}\ket 0=\ket 0$ if the last symbol is $0$ rejecting the input, or $U_{(1,1)}\ket 0=\ket 1$ if the last symbol is $1$ accepting the input.
\end{proof}

Restricting the decider to a constant function accepting any input, a CTQA can still
recognize a non-regular language, as stated by the following theorem. Let
$\texttt L_{a\sim b}^\lambda=\{x~|~ |x|_a=n, |x|_b=m,
\cos^2(\frac{\pi(n-m)}{2(n+m)})\geq\lambda\}$. Using a pumping argument, it can
be proved that $\texttt L_{a\sim b}^\lambda$ is not regular. Tadaki, Yamakami and Lin~\cite{TYL10} showed that the language $\texttt{L}_{a\sim b}=\{ |x|~|~ |x|_a=|x|_b=n,n \in \mathbb{N}\}$ is not in $\mathrm{REG}/n$, where $\mathrm{REG}/n$ is the class of languages recognized by deterministic finite automata assisted by advice. This fact also implies that $\texttt L_{a\sim b}^\lambda \notin \mathbf{REG}/n$. (Note that when $\lambda=1$, $\texttt L_{a\sim b}^\lambda=L_{a\sim b}$.)

\begin{theorem}\label{the:a-simb-b}
  For any $\lambda\in[0,1]$,
  $\mathbf{\Sigma^*}\mbox{-}\mathbf{CTQ}_\lambda \nsubseteq \mathbf{REG}/n$.
  \end{theorem}
  \begin{proof}
    Since $\texttt L_{a\sim b}^\lambda \notin  \mathbf{REG}/n$, it suffices to prove
  that $\texttt L_{a\sim b}^\lambda \in \mathbf{\Sigma^*}\mbox{-}\mathbf{CTQ}_\lambda$.
  It suffices to construct an automaton $M'$ similar to the automaton $M$ from the proof of Theorem~\ref{the:anbn} with a decider $D'$ defined by
  $D'(x)=1$, for any $x\in\{a,b\}^*$. Indeed, on input $x\in \texttt L_{a\sim
  b}^\lambda$ the machine $M'$ will execute the operator $U_{(a,\frac 1{n+m})}$
  $n$ times and the operator $U_{(a,-\frac 1{n+m})}$ $m$ times, in any order,
  producing  the quantum state
\[
\cos\left(\frac{\pi(n-m)}{2(n+m)}\right)\ket 0
+
i\sin\left(\frac{\pi(n-m)}{2(n+m)}\right)\ket 1.
\]

The probability of accepting a string in $\texttt L_{a\sim b}^\lambda$ is $\cos^2(\frac{\pi(n-m)}{2(n+m)})$ which is at least $\lambda$. If $x\notin \texttt L_{a\sim b}^\lambda$, then the probability of accepting $x$ is less than $\lambda$. 
  \end{proof}

It can be argued that the time-schedule demands too much precision to be
implemented. Indeed, running a unitary operator for time $\frac 1n$ with large
$n$ may be
a challenge. Fortunately, the time can be re-scaled as stated by Theorem~\ref{the:scale} below.

For any input $x$, time-schedule $\tau(x)=(\tau_1 \cdots \tau_n)$ and a positive real number $k$, we say that $k\tau(x)=(k\tau_1 \cdots k\tau_n)$ is the time-schedule $\tau$ \emph{scaled} by $k$. 
\begin{theorem}\label{the:scale}
Given any positive real constant $k$, for any CTQA $M$ with time-schedule $\tau$, there exists a CTQA $M'$ with time-schedule $\tau'$ where $\tau'$ is $\tau$ scaled by $k$ and $M'$ recognizes the same language as $M$.
\end{theorem}
\begin{proof}
Let $M=(Q,\Sigma,\{H_\sigma\}_{\sigma\in\Sigma},\tau,s,A,R)$ be a CTQA and $S$
an scheduler that computes a time-schedule $\tau(x)=(\tau_1 \cdots \tau_{|x|})$
for $M$. Then, we can define
$M'=(Q,\Sigma,\{H_\sigma'\}_{\sigma\in\Sigma},\tau',s,A,R)$ where for each
$\sigma\in \Sigma$, the 
Hamiltonian $H'_\sigma$ is $H'_\sigma = \frac{H_\sigma}k$ and
  $\tau'(x)=(k\tau_1 \cdots k\tau_{|x|})$. Therefore, the transition operators $U_{(\sigma',t)}$ are defined as
   $U'_{(\sigma,t)}=U_{(\sigma,\frac tk)}$.
On input $x=x_1\cdots x_n$, the machine $M$ computes
  \[
  U_{(x_1,\tau_1)}\cdots U_{(x_n,\tau_{n})}\ket s
  =
  U'_{(x_1,k\tau_1)}\cdots U'_{(x_n,k\tau_{n})}\ket s
  \]
which is also the computation done by $M'$.
\end{proof}

To end this section, notice that for most of the proofs up to this point, we only required 2 inner states of a CTQA. Not only a CTQA assisted by an external scheduler helps in recognizing more languages, but also this fact demonstrate the succinctness of the model.

\section{Cutpoint-Union of Languages recognized by CTQAs}\label{sec:operations}
\subsection{The Cutpoint-Union Operation}

We define the \emph{cutpoint-union} of two cutpoint languages as follows. 
Let $L_1^{\lambda_1}$ and $L_2^{\lambda_2}$ be two cutpoint languages over the same alphabet $\Sigma$.
That is, there exists $M_1$ and $M_2$ such that $L_1^{\lambda_1} =
\{x\in\Sigma^* : p_{M_1,A}(x)\geq \lambda_1\}$ and  $L_2^{\lambda_2} =
\{x\in\Sigma^* : p_{M_2,A}(x)\geq \lambda_2\}$. The cutpoint-union
$L_1^{\lambda_1}\doublecup L_2^{\lambda_2}$ is the union of the languages,
where the cutpoint is taken to be the probability of the union of the accepting
events. That is, there exists a machine $M$ such that
\[
  L_1^{\lambda_1}\doublecup L_2^{\lambda_2}=
  \{x\in\Sigma^*\mid
  p_{M,A}(x)\geq p_{M_1,A}(x)+p_{M_2,A}(x)-p_{M_1,A}(x)p_{M_2,A}(x)\}.
\]

Now we prove that if $M_1$ and $M_2$ are CTQAs with deciders agreeing on all
inputs, then we can construct another CTQA with a rotating tape head
recognizing the cutpoint-union of the languages.

\subsection{\texorpdfstring{$k$}{k}-Rotating MCQFAs and Language Closure}
A \emph{$k$-rotating Moore-Crutchfield quantum finite automaton} ($k$MCQFA),
with $k = 2^n$, with $n\geq 1$, is a
variation of a MCQFA where the automaton is equipped with a circular tape with
endmarkers and a unary counter up to $k$ named ``K''. The input tape is scanned  exactly $k$ times with a partial
measurement over the counter every time the left marker is scanned, and a global
measurement at the end. More formally, let $\hat{\Sigma}=\Sigma\cup
\{\vdash,\dashv\}$ where the symbols $\vdash$ and $\dashv$ are the left and
right endmarkers respectively. We define a $k$MCQFA $M=(Q\times C,\hat
\Sigma,\{U_\sigma\}_{\sigma\in\Sigma},s,A,R)$ where
$C=\{0,1\}^n$.

In order to incorporate a time-schedule into the aforementioned $k$MCQFA we define a $k$CTQA $M=(Q,\hat \Sigma,\{H_\sigma\}_{\sigma\in\Sigma},\tau,s,A,R)$ where
$\tau$ maps $x$ to a sequence of $k|x|+2k$ rational real
numbers of the form
\begin{align}\nonumber
  (&0, t_{(1,1)},\dots,t_{(1,|x|)}, 0,\\\nonumber
    &0,t_{(2,1)},\dots,t_{(2,|x|)},0,\\\nonumber
    &\dots,\\
  &0,t_{(k,1)}\dots,t_{(k,|x|)},0).
  \label{eq:rotation-schedule}
\end{align}
That is, the writer always writes a schedule as in
Eq.~\eqref{eq:rotation-schedule} in order to consider the left and right
endmarkers and deliver a (possibly different) time-schedule for each sweep of
the input tape head.

The Hilbert space of $M$ is $\mathbb C^{Q\times C}$.
Let $\ket{\psi_i,r}$ be a quantum state of $M$ in step $i$ scanning $x_j\in\Sigma$ at the $r$-th sweep. Then we have that
\[
\ket{\psi_{i+1},r}=U_{(x_j,t_{(r,j)})}\ket{\psi_i,r}
\]
where $U_{(x_j,t_{(r,j)})}$ does not modify the counter K.

Every time the left endmarker is encountered, the transition operator of $M$ acts
as the identity.
When a right endmarker is encountered, the partial measurement is defined by the
collection of projectors $\{M_c\}_{c\in\{0,1\}^n}$ defined by
$I\otimes\ketbra{c+1}c$, where $+$ is the sum modulo $2^n$. When the
output of such a
measurement is $k$, then a global measurement with the observable $E_N\oplus E_A\oplus
E_R$ is performed, where
\begin{itemize}
\item $E_A=span\{\ket{q,k}\mid q\in A\}$,
\item $E_R=span\{\ket{q,k}\mid q\in R\}$, and
\item $E_N=span\{\ket{q,k}\mid q\in\overline{A\cup R}\}$.
\end{itemize}

\begin{theorem}\label{the:union-c}
  Let $\mathcal C$ be any complexity class, $\texttt{L}_1\in \mathcal{C}\mbox{-}\mathbf{CTQ}_{\lambda_1}$, and $\texttt{L}_2 \in \mathcal{C}\mbox{-}\mathbf{CTQ}_{\lambda_2}$. If the deciders for $\texttt{L}_1$ and $\texttt{L}_2$ agree on all inputs, then there exists a 2CTQA that recognizes $\texttt L_1^{\lambda_1} \doublecup \texttt L_2^{\lambda_2}$ with cutpoint $\lambda=\lambda_1+\lambda_2-\lambda_1\lambda_2$.
\end{theorem}
\begin{proof}
  Let
  \begin{align*}
    M_1&=(Q_1,\Sigma,\{H^1_\sigma\}_{\sigma\in\Sigma},\tau_1,s_1,A_1,R_1)\text{
    and}\\
    M_2&=(Q_2,\Sigma,\{H^2_\sigma\}_{\sigma\in\Sigma},\tau_2,s_2,A_2,R_2)    
  \end{align*}
  be CTQAs recognizing $\texttt{L}_1$ and $\texttt{L}_2$, respectively. We define a $2$CTQA $M$ recognizing $\texttt{L}=\texttt{L}_1\cup\texttt{L}_2$.

Define $M=(Q\times C,\hat \Sigma,\{H_\sigma\}_{\sigma\in\Sigma},\tau,s,A,R)$ where
\begin{itemize}
\item $Q=Q_1\times Q_2$,
\item $C=\{0,1\}$,
\item $A=(A_1\times Q_2)\cup (Q_1\times A_2)$,
\item $R=R_1\times R_2$,
\item $s=(s_1,s_2)$,
\item if $\tau_1(n)=(\tau_{11} \cdots \tau_{1n})$ and
  $\tau_2(n)=(\tau_{21} \cdots \tau_{2n})$, then
  \[
    \tau(n)=(0,\tau_{11} \cdots \tau_{1n},0,0,\tau_{21} \cdots \tau_{2n},0).
  \]
\item The Hamiltonians $H_\sigma$ for $\sigma\in \Sigma$ (with no
endmarkers)  are defined as $H_\sigma = \frac{-\log(U_{(\sigma,t)})}{it}$, where
$U_{(\sigma,t)}$ are defined as
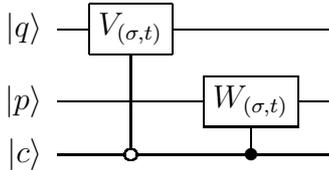
\begin{center}
\large
\parbox{\textwidth}{
  \Qcircuit @C=1em @R=.7em{
  &&\lstick{\ket q} & \gate{V_{(\sigma,t)}} & \qw & \qw\\
  &&\lstick{\ket p}& \qw & \gate{W_{(\sigma,t)}} & \qw\\
  &&\lstick{\ket c}& \ctrlo{-2} & \ctrl{-1} & \qw
}
}
\end{center}
where $q\in Q_1$, $p\in Q_2$, $c\in C$, and $V_{(\sigma,t)}$ is calculated from
$H_\sigma^1$ and $W_{(\sigma,t)}$ from $H_\sigma^2$.
\end{itemize}

Now we consider the construction of a scheduler $S=(D,W)$ for $M$. Let $S_1=(D_1,W_1)$ and $S_2=(D_2,W_2)$ be the schedulers for $M_1$ and $M_2$ respectively. Since the deciders $D_1$ and $D_2$ agree on all inputs, that is $D_1(x)=D_2(x)$ for all $x\in \Sigma^*$, we let $D=D_1$. Then the writer $W$ on input $D(x)$ and $|x|$ constructs a time-schedule 
\[
(0,W_1(D(x),|x|),0,0,W_2(D(x),|x|),0).
\]

Let $x\in \texttt{L}_1\cup \texttt{L}_2$. After $M$ makes two sweeps of an
input tape,  note that there is never any entanglement between subspaces
$\mathbb{C}^{Q_1}$ and $\mathbb{C}^{Q_2}$. Thus, a final quantum state
(disregarding the counter) is of the form
\begin{align}\nonumber
  \sqrt{p_1}\sqrt{p_2} & \ket{rs} \\\nonumber
  +\sqrt{p_1}\sqrt{1-p_2} & \ket{r\bar{s}}\\\nonumber
  +\sqrt{1-p_1}\sqrt{p_2}&\ket{\bar{r}s}\\
  +\sqrt{1-p_1}\sqrt{1-p_2}&\ket{\bar{r}\bar{s}},\label{eq:final-state}
\end{align}
where $\ket r$ and $\ket s$ are states in $\mathbb{C}^{A_1}$ and $\mathbb{C}^{A_2}$ respectively, and $\ket{\bar{r}}$ and $\ket{\bar{s}}$ are non-accepting inner states of $\mathbb{C}^{Q_1-A_1}$ and $\mathbb{C}^{Q_2-A_2}$ respectively.

The accepting probability of $M$ is 
\begin{align*}
  1-(\sqrt{1-p_1}\sqrt{1-p_2})^2
  = p_1+p_2-p_1p_2
  \geq \lambda_1+\lambda_2-\lambda_1\lambda_2.
  \tag*{\qedhere}
\end{align*}
\end{proof}

The following corollary shows that for the case when the scheduler of both CTQAs $M_1$ and $M_2$ are the same, then a 1CTQA suffices to recognize the cutpoint-union of two languages.

\begin{corollary}\label{cor:union-scheduler}
  Let $\texttt{L}_1^{\lambda_1}$ and
  $\texttt{L}_2^{\lambda_2}$. If the CTQAs for $\texttt{L}_1$ and $\texttt
  L_2$ have the same scheduler $S$, then there exists a CTQA with scheduler $S$
  that recognizes $\texttt{L}_1^{\lambda_1}\doublecup \texttt{L}_2^{\lambda_2}$ with cutpoint $\lambda=\lambda_1+\lambda_2-\lambda_1\lambda_2$.
\end{corollary}
\begin{proof}
  The construction of a CTQA $M$ recognizing
  $\texttt{L}_1^{\lambda_1}\doublecup \texttt{L}_2^{\lambda_2}$ is similar to
  the proof of Theorem~\ref{the:union-c} with the exception that there is no
  extra bit to remember the stage of the sweep of an input tape and the
  definition of each $H_\sigma$. Since the scheduler for $\texttt L_1$ and
  $\texttt L_2$ are the same, we can define the Hamiltonians of $M$ as
  $H_\sigma=H^1_\sigma\oplus H^2_\sigma$, that is, the direct sum of the
  individual Hamiltonians, and so, the transition operators are
  $U_{(\sigma,t)}=V_{(\sigma,t)}\otimes W_{(\sigma,t)}$. Hence, the probability
  that $M$ accepts $x\in \texttt{L}_1^{\lambda_1}\doublecup
  \texttt{L}_2^{\lambda_2}$ is at least
  $\lambda_1+\lambda_2-\lambda_1\lambda_2$.
\end{proof}

\section{Concluding Remarks and Open Problems}\label{sec:open}
{W}e introduce a new model of quantum computation with classical
control called CTQA (for classically time-controlled quantum automata). {In a CTQA} an
infinite number of unitary operators are defined from the evolution of
Hamiltonians associated to each symbol and their time executions are externally
controlled by a scheduler.

We show in Theorem~\ref{thm:Halting} that if Moore-Crutchfield quantum automata
use  unitary operators issued from the evolution times of Hamiltonians with
unrestricted time-schedules, then they can recognize non-recursive languages.
If the scheduler is defined via a finite-state automaton, however, a CTQA can
recognize non-regular languages, as shown {by} Theorems~\ref{the:anbn}
and~\ref{the:a-simb-b}.

In Figure~\ref{fig:diagram}, we describe the containment relationships between the classes that we have shown in this paper. A dashed arrow indicates non-inclusion, while a solid arrow indicates inclusion. Additionally, in Theorem~\ref{the:scheduler-power}, we prove that a $\mathcal C-\mathbf{BCTQ}_0$ has at least as much power as its decider. In Theorem~\ref{the:bounded-error}, we show that a non-trivial promise language can be recognized by a CTQA with bounded error, and in Theorem~\ref{the:scale}, we show that the time-schedule does not need to consider arbitrarily large real numbers. Finally, in Theorem~\ref{the:union-c}, we show that the cutpoint-union of two languages recognized by CTQAs can be recognized by a 2CTQA, and in Corollary~\ref{cor:union-scheduler}, we show that if the schedulers of two CTQAs are the same, then a 1CTQA suffices to recognize the cutpoint-union of the languages.

\begin{figure}[h]
\centering
\begin{tikzpicture}[x=3cm,y=2cm]
\node (CTQ) at (0,0) [draw, rectangle, rounded corners] {$\mathbf{BCTQ}$ = ALL};
\node (SIGMA_CTQ_l) at (1,-1) [draw, rectangle, rounded corners] {$\mathbf{\Sigma^*}\mbox{-}\mathbf{CTQ}_\lambda$};
\node (KWQFA_CTQ_l) at (0,-2) [draw, rectangle, rounded corners] {$\mathbf{KWQFA}\mbox{-}\mathbf{CTQ}_\lambda$};
\node (Sigma_BCTQ_0) at (-1,-2) [draw, rectangle, rounded corners] {$\mathbf{\Sigma^*}\mbox{-}\mathbf{BCTQ}_0$};
\node (REG_n) at (1,-2) [draw, rectangle, rounded corners] {$\mathbf{REG}/n$};
\node (REG) at (0,-3) [draw, rectangle, rounded corners] {$\mathbf{REG}$};
\node (KWQFA) at (0,-4) [draw, rectangle, rounded corners] {$\mathbf{KWQFA}$};
\node (MCQFA) at (-1,-5) [draw, rectangle, rounded corners] {$\mathbf{MCQFA}$};
\node (ti_BCTQ) at (1,-5) [draw, rectangle, rounded corners] {$\mathbf{ti}\mbox{-}\mathbf{BCTQ}$};
\draw[<->] (ti_BCTQ) -- (MCQFA) node[midway, above] {Thm.~\ref{thm:time-independent}};
\draw[->] (MCQFA) -- (KWQFA) node[midway,above,sloped] {\cite{AF98}};
\draw[->] (KWQFA) -- (REG) node[midway,right] {\cite{AF98}};  
\draw[->, dashed] (KWQFA) -- (Sigma_BCTQ_0) node[midway,below,sloped] {Thm.~\ref{thm:One-CTQA}};
\draw[->, dashed] (REG) -- (KWQFA_CTQ_l) node[midway,left] {Thm.~\ref{the:anbn}};
\draw[->] (REG) -- (REG_n) node[midway,below,sloped] {\cite{TYL10}};
\draw[->, dashed] (REG_n) -- (SIGMA_CTQ_l) node[midway,right] {Thm.~\ref{the:a-simb-b}};
\draw[->] (Sigma_BCTQ_0) -- (CTQ) node[midway,above,sloped] {Thm.~\ref{thm:Halting}};
\draw[->] (KWQFA_CTQ_l) -- (CTQ) node[pos=0.3,left] {Thm.~\ref{thm:Halting}};
\draw[->] (REG_n) -- (CTQ) node[midway,sloped,below] {Thm.~\ref{thm:Halting}};
\draw[->] (SIGMA_CTQ_l) -- (CTQ) node[midway,sloped,above] {Thm.~\ref{thm:Halting}};
\end{tikzpicture}
\caption{Containment relationships between classes. Solid arrows represent inclusion, while dashed arrows represent non-inclusion.}
\label{fig:diagram}
\end{figure}

The CTQA model is an interesting model to study quantum
computations assisted by a classical control that can tune or adjust execution
times. Below we present some interesting open problems that remain from this
work.
\begin{enumerate}
\item \emph{Upper bound for classes of languages recognized by CTQAs}. To prove
  an upper bound on the simulation of CTQAs we require a simulation of the
  behavior of schedulers. Since a scheduler outputs real numbers, it is necessary
  to consider computable real numbers and study how much error in the
  time-dependent computation is introduced. Notice however that all the
    results given in this paper use schedulers outputting just fractional numbers.
\item \emph{Closure of well-known operations}. We showed a sufficient condition for the closure of CTQAs under cutpoint-union (Corollary~\ref{cor:union-scheduler}). It would be interesting to continue this line of study to find new conditions for the closure of CTQAs under union, as well as under other operations such as intersection, homomorphism, inverse homomorphism, etc., in a bounded-error setting with cut points.
\item \emph{Impossibility results}. It would be interesting to see a lower bound technique for CTQAs analogous to a pumping lemma in order to obtain some impossibility result.
\item \emph{Recognition power of $k$MCQFAs}. In order to study the cutpoint-union of two languages that are recognized by CTQAs, we introduced a new idea of rotating automata or $k$MCQFA. Whether the class of languages recognized by $k$MCQFAs (with no scheduler) equals the class of languages recognized by MCQFAs is an open and interesting problem.
\item \emph{Simulation of advised quantum computation.} As we mentioned in Section~\ref{sec:definition}, simulation of advised quantum automata using CTQAs might require some new ideas for Hamiltonian interpolation. We do not believe that interpolation is the only way to simulate advised computation, however, developing techniques on how to simulate advice using time-dependent Hamiltonians can be an interesting research subject.
\item\emph{Computational restrictions on the writer.} When we incorporated
  computational restrictions to a scheduler, only the decider was affected and
  the writer was left untouched. We can, however, also restrict the
  computational resources of the writer. This can lead to interesting new
  research questions like how much precision a time-schedule needs in order to
  do reliable computations.
\end{enumerate}


\section*{Acknowledgments}
We wish to thank Ariel Bendersky, Federico Holik, and Malena Ivnisky for their
comments on an earlier version of this paper.

\bibliographystyle{abbrvnat}
\bibliography{biblio}

\appendix
\section{Proof of the Non-Regularity of \texorpdfstring{$\texttt{L}_{ab}^\lambda$ for any $\lambda>0$}{Lablamb for any lamb>0}.}\label{ap:non-regular}
To prove the non-regularity of $\texttt{L}_{ab}^\lambda$ we use a pumping argument. Recall that the pumping lemma for regular languages states that if $\texttt{L}$ is regular, then there exists an integer $p\geq 1$ (the pumping length) such that any string $w\in \texttt{L}$ can be written $w=xyz$ where (1) $|y|\geq 1$, (2) $|xy|\leq p$, and (3) for all $n$ it holds $xy^nz\in \texttt{L}$.

Suppose that $\texttt{L}_{ab}^\lambda$ is regular and consider the string $w=a^pb^p\in \texttt{L}_{ab}^\lambda$, where $p$ is the pumping length. Then, for some $i\geq 1$, we have that $y=a^i$, $x=a^{p-i}$, and $z=b^p$.

Now consider the string $a^{p-i}a^{in}b^p$ for $n\geq 1$. Since $a^pb^p\in \texttt{L}_{ab}^\lambda$, it holds that $\cos^2\left(\frac{\pi(in-i)}{2(2p-i+in)}\right)\geq \lambda$ when $n=1$. When $n>1$, however, note that
\[
\lim_{n\to \infty}\frac{in-i}{2p-i+in}=1,
\]
and hence, there exists $n_0>1$ such that $\cos^2\left(\frac{\pi(in_0-i)}{2(2p-i+in_0)}\right)<\lambda$. Therefore, the string $xy^{n_0}z=a^{p-i}a^{in_0}b^p$ is not in  $\texttt{L}_{ab}^\lambda$ and cannot be pumped. This is a contradiction with the pumping lemma, and thus, it implies that $\texttt{L}_{ab}^\lambda$ is not regular. Note that this argument does not work with $\lambda=0$.

\end{document}